\documentclass[final,1p,times]{elsarticle}

\usepackage{amssymb}
 \usepackage{amsthm}
\newtheorem{Theorem}{Theorem}[section]
\newtheorem{Definition}{Definition}[section]
\newtheorem{lemma}{Lemma}[section]
\newtheorem{Corollary}{Corollary}[section]
\newtheorem{Fact}{Fact}[section]

\journal{Information Processing Letters}

\begin{document}

\begin{frontmatter}



\title{On the Complexity of Edge Packing and Vertex Packing}

\author{Parvathy K. N.}
\ead{paranns@gmail.com}
\author{˜Sameera Muhamed Salam}
\ead{shemi.nazir@gmail.com}
\author{Sudeep K. S.}
\ead{sudeep@nitc.ac.in}
\author{K. ˜Murali Krishnan}
\ead{kmurali@nitc.ac.in}
\address{Department of Computer Science and Engineering, National Institute of Technology Calicut, Kerala, India 673601}
\begin{abstract}
This paper studies the computational complexity of the Edge Packing problem and the Vertex Packing problem.  The edge packing 
problem (denoted by $\overline{EDS}$) and the vertex packing problem (denoted by $\overline{DS} $) are 
linear programming duals of the edge dominating set problem and the dominating set problem respectively.  
It is shown that these two problems are equivalent to the set packing problem with respect to 
hardness of approximation and parametric complexity.  It follows that $\overline{EDS}$
and $\overline{DS}$ cannot be approximated asymptotically within a factor of
$O(|V|^{\frac{1}{2}-\epsilon})$ for any $\epsilon>0$ unless $NP=ZPP$ where, $V$ is the set of vertices of the given graph.   
This is in
contrast with the fact that the edge dominating set problem is 2-approximable whereas the dominating set problem
 is known to have an
$O(\log$ $|V|)$ 
approximation algorithm.
It also follows from our proof that $\overline{EDS}$ and $\overline{DS}$ are $W[1]$-complete under standard parameterization. 
\end{abstract}


\end{frontmatter}


\section{Introduction}

The hardness of approximation and parametric complexity of the dominating set problem and the edge dominating 
set problem are well 
studied in literature.  In this paper we 
investigate the approximability and parametric complexity of linear programming (LP) duals of these problems.

Given a graph $G=(V,E)$, a set of edges $M\subseteq E$ is called
an \textit{edge dominating set} (EDS) if every edge in $E \setminus M$
is adjacent to at least one edge in $M$.  The minimum edge dominating set problem
asks for such subset $M$ of minimum cardinality.  It is NP-complete
even for planar and
bipartite graphs \cite{yannakakis1980edge}.  
It is also known that this problem is 2-approximable \cite{yannakakis1980edge, chlebik2003approximation} and
fixed parameter tractable \cite{fernau2006edge}, where the
parameter is the size of the edge dominating set (called the standard parameterization).

 \textit{Dominating set} (DS) of a graph $G=(V,E)$ is a subset $V'\subseteq V$
such that for every $u\in V\setminus V'$, there exists at least one $v\in V'$ for
which $(u,v)\in E$.  A minimum dominating set is a dominating set
of minimum cardinality.  The dominating set problem is
proved to be equivalent to the set cover problem (SC) by Kann \cite{kann1992approximability} and hence is $1+\log(|V|)$
approximable \cite{feige1998threshold, chlebik2004approximation}.  With respect to parametric complexity 
it is W[2]-complete under standard parameterization\cite{cesati2001compendium}.   
 
This paper studies the edge packing problem ($\overline{EDS}$) and the vertex packing problem ($\overline{DS}$) 
which are the LP duals of EDS and DS respectively.  We show
that these problems are equivalent to the set packing problem.
  Consequently, edge packing problem is $\sqrt{|E|}$ approximable
 and vertex packing problem is $\sqrt{|V|}$ approximable where,  $E$ and 
$V$ are the set of edges and vertices respectively in the given graph.  Moreover, they cannot be approximated 
asymptotically within a factor of $O(|V|^{\frac{1}{2}-\epsilon})$ for
any $\epsilon>0$ unless $NP=ZPP$.
As all reductions we propose are parameter preserving, they are 
FPT reductions \cite{chen1993fixed} and conclude that the edge packing problem and the vertex packing problem are W[1]-complete.
 Section 2 introduces the problems discussed in this paper  
and Section 3 presents the reductions.  
\section{Preliminaries}
Let $G=(V,E)$ be a graph.  Let $e\in E$ and $v\in V$.  Define $N(e)$=\{$e$\} $\cup$ \{$e^\prime \in E : e'$  \textit{share 
an endpoint with $e$}\}  
and $N(v)$= \{$v$\} $\cup$ \{$v'\in V:(v',v)\in E$ \}.

\begin{Definition}[$\overline{DS}$]

Given a graph $G=(V,E)$, the Vertex Packing problem ($\overline{DS}$) is to find a subset $V'$ of $V$ of maximum cardinality
 such that no two vertices in $V'$ have a common neighbour in $V$.  That is, 
if $v_{i},v_{j}\in V'$, then $N(v_{i})\cap N(v_{j}) = \phi$.
\end{Definition}

\begin{Definition}[$\overline{EDS}$)]
Given a graph $G=(V,E)$, the Edge Packing problem ($\overline{EDS}$) is to find a subset $E'$ of $E$ of maximum cardinality 
such that no two edges in $E'$ have a common neighbour in $E$.  That is, if $e_{i},e_{j}\in E'$, then $N(e_{i})\cap N(e_{j}) = \phi$. 
\end{Definition}
\begin{Definition}[SC]
Given a universal set $U=\{a_1,a_2,.....a_n\}$ and a collection of subsets $S=\{s_{1},s_{2},....,s_{m}\}$
of $U$, the set cover problem is to find a subset $S'\subseteq S$ of minimum cardinality such that every element
in $U$ belongs to at least one member of $S'$.
\end{Definition}

\begin{Definition}[$\overline{SC}$]
Given a universal set $U=\{a_1,a_2,.....a_n\}$ and a collection of subsets $S=\{s_{1},s_{2},....s_{m}\}$, the problem 
$\overline{SC}$ requires to find a set $U' \subseteq U$ of maximum cardinality such that every element in 
$U'$ occurs in at most one member of $S$.

\end{Definition}
\begin{Definition}[Set packing problem]

Given a universal set $U=\{a_1,a_2,.....a_n\}$ and a collection of subsets $S=\{s_{1},s_{2},....,s_{m}\}$
of $U$, the set packing problem is to find a collection of disjoint sets $S'\subseteq S$ of maximum cardinality.\end{Definition}

It is easy to see that $\overline{DS}$, $\overline{EDS}$ and $\overline{SC}$ are LP duals of the dominating
set problem, the edge dominating set problem and the set cover problem respectively.  The following is known about the 
set packing problem.
\begin{Fact} \label{lema:11}
    \begin{enumerate}
\item
  The set packing problem is $\sqrt{|U|}$ approximable \cite{crescenzi1998compendium}.
\item For any $\epsilon>0$, the set packing problem is not approximable within a factor $|U|^{\frac{1}{2}-\epsilon}$ 
unless $P=NP$ and $|S|^{1-\epsilon}$ unless $NP=ZPP$  \cite{halldórsson2000independent}.
\item The set packing problem is W[1]-complete under standard parameterization\cite{cesati2001compendium}.
 \end{enumerate}
\end{Fact}
It is not hard to see that the LP dual of the set cover problem ($\overline{SC}$) and 
the set packing problem are equivalent.
We sketch a proof for this fact below for the sake of completeness.

\begin{Theorem} \label{theorem:1}
 $\overline{SC}$ and the set packing problem are equivalent.
\end{Theorem}


 
\begin{proof}
\textit{$\overline{SC}$ is reducible to the set packing problem:} Given an $\overline{SC}$ instance $T=(U,S)$ where, 
$U=\{a_{1},a_{2},..,a_{n}\}$
and a collection of subsets $S=\{s_{1},s_{2},....s_{m}\}$, we can construct
a set packing instance $T'=(U',S')$ such that $U'=\{s_{1},s_{2},..,s_{m}\}$
and $S'=\{s'_{1},s'_{2},....s'_{n}\}$ where, 
$s'_{i}=\{s_{j}:a_{i}\in s_{j}, 1\leq j\leq m\},$ for each $1\leq i\leq n$. 

Any set packing in $T'$ can be converted to a feasible solution
for $\overline{SC}$ in $T$ of the same size: If $S_1=\{s'_{1},s'_{2},...,s'_{k}\}$ is a set packing of 
size $k$ in $T'$, then 
$S_2=\{a_1,a_2,....a_k\}$ will be a feasible solution for $\overline{SC}$ of size $k$ in $T$.  
Consider any two arbitrary elements 
$s'_{i},s'_{j}$ from $S_1$.  This implies that $s'_{i}\cap s'_{j}=\phi$; which in turn implies that $a_i$ and $a_j$
do not simultaneously occur in any of the subsets in $S$ (by construction).  So it is safe to add $a_i$ and $a_j$ to  
$S_2$.  Hence $S_2$ is a feasible solution for $\overline{SC}$ of size $k$ in $T$. 

Similarly, any feasible $\overline{SC}$ in $T$ can be converted to a feasible set packing in $T'$ of same size: 
If $S_1=\{a_{1},a_{2},...,a_{k}\}$ is 
feasible for $\overline{SC}$ in $T$ with $|S_1|=k$, then $S_2=\{s'_1,s'_2,....s'_k\}$ will be a set packing of size $k$ 
in $T'$.  Consider any two arbitrary elements $a_i,a_j$ from $S_1$. This implies $a_i$ and $a_j$
do not simultaneously occur in any of the subsets in $S$ which in turn implies that $s'_{i}\cap s'_{j}=\phi$ 
in $S'$ (by construction). So it is safe to add $s'_i$ and $s'_j$ to $S_2$.  Hence $S_2$ is a set packing of size 
$k$ in $T'$.  Thus $\overline{SC}$ is reducible to the set packing problem. 

\textit{The set packing problem is reducible to $\overline{SC}$:} Given a set packing instance $T=(U,S)$ 
where, $U=\{a_{1},a_{2},..,a_{n}\}$
and a collection of subsets $S=\{s_{1},s_{2},....s_{m}\}$, we can construct
an $\overline{SC}$ instance $T'=(U',S')$ such that $U'=\{s_{1},s_{2},..,s_{m}\}$
and $S'=\{s'_{1},s'_{2},....s'_{n}\}$ where, 
$s'_{i}=\{s_{j}:a_{i}\in s_{j}, 1\leq j\leq m\}$, for each $1\leq i\leq n$. 

Any set packing in $T$ will be feasible for $\overline{SC}$ in $T'$: 
Let $S_1=\{s_{1},s_{2},...,s_{k}\}$ be a feasible set packing in $T$ with $|S_1|=k$ and suppose it is not feasible 
for $\overline{SC}$ in $T'$.  
Then there are at least two elements, $s_i$ and $s_j$ in $S_1$ such that $s_i,s_j \in s'_p$ for some $p$ 
$(1\leq p\leq n)$, which implies that $p \in s_i$ and $p \in s_j$ (by construction).  This again implies that
 $p \in s_i\cap s_j$ and $S_1$ is not a set packing in $T$.  Hence contradiction.

Any feasible solution for $\overline{SC}$ in $T'$ will be a set packing in $T$: Let  $S_2=\{s_1,s_2,....s_k\}$ be feasible 
for $\overline{SC}$ in $T'$ with $|S_1|=k$ and suppose it is not a set packing in $T$.  
This implies that there is some $s_i,s_j$ in 
$S_2$ such that $p \in s_i\cap s_j$ for some $p$ $(1\leq p\leq n)$ which again implies that $s_i,s_j \in s'_p$ 
(by construction).  
So $S_2$ is not a feasible solution for $\overline{SC}$ in $T'$.  Hence contradiction. 
\end{proof}
Applying  the Fact~\ref{lema:11} and Theorem~\ref{theorem:1}, we get the
following corollaries.  Let $T=(U,S)$ be an instance of  $\overline{SC}$.

\begin{Corollary}\label{Coro:1}

For any $\epsilon>0$, $\overline{SC}$ is not approximable within a factor 
$|S|^{\frac{1}{2}-\epsilon}$ unless $P=NP$ and $|U|^{1-\epsilon}$ unless $NP=ZPP$.  

\end{Corollary}

\begin{Corollary}\label{Coro:2}
 $\overline{SC}$ is $\sqrt{|S|}$ approximable.
\end{Corollary}

We will be using the following theorem to prove the equivalence 
of $\overline{SC}$, $\overline{EDS}$ and $\overline{DS}$.
\begin{Theorem} \label{thm:1}
(Hall \cite[p.106]{harris2008combinatorics}) Let $G$ be a bipartite graph with
partite sets $X$ and $Y$.  $X$ can be matched into $Y$ iff $|S|\leq|N(S)|$ for all subsets
$S$ of $X$ (Here $N(S)=\{ y\in Y|(x,y)\in E(G), \forall x\in S\}$).
\end{Theorem}

\section{Equivalence of $\overline{SC}$, $\overline{EDS}$ and
$\overline{DS}$}
\begin{Theorem}\label{thm:2}
 $\overline{SC}$, $\overline{EDS}$ and $\overline{DS}$ are equivalent with
respect to hardness of approximation.
\end{Theorem}

\begin{proof}
 We split the proof into four parts.  The consequences of intermediate results
are noted as corollaries.



\textbf{Part (i)} \textit{Reduction from $\overline{SC}$ to $\overline{EDS}$:}
Given an $\overline{SC}$ instance $T=(U,S)$ where,
$U=\{a_{1},a_{2},.....a_{n}\}$ and
$S=\{s_{1},s_{2},....s_{m}\}$. Construct
a graph $G=(V,E)$ where, 
$V=U\cup S$ and $E=\{(a_{i},a_{j})|\exists s_{p}\in S:a_{i}\in s_{p}$ and $ a_{j}\in s_{p}\}\cup\{(a_{i},s_{p})|a_{i}\in s_{p}\}$.

Let $U'=\{a_1,a_2,...a_k\}$ be a solution of size $k$ for $\overline{SC}$ in $T$.  Let $N(U')=\{ v' \in V:(u',v')\in E, 
\forall u'\in U'\}$ 
Consider the subgraph $G'=(V',E')$ of $G$ where, $V'={U'\cup\{N(U')\cap S\}}$ and 
$E'=\{(a_i,s_j)\in E|a_i\in U';s_j \in N(U')\cap S\}$.

\begin{lemma}For any $U''\subseteq U'$, $|U''|\leq|N(U'')|$ in $G'$.\end{lemma}
\begin{proof}Suppose $ |U''|>|N(U'')|$, then $\exists a_i,a_j,s_q$ such that $(a_i,s_q)\in E'$ and $(a_j,s_q)\in E'$ 
(pigeon hole principle) which means $U'$ is not a packing in $T$. Hence contradiction.
\end{proof}

So for any $U''\subseteq U'$, $|U''|\leq|N(U'')|$ in $G'$.  Now, by Theorem~\ref{thm:1}, 
there exists a complete matching saturating every vertex of 
$U'$ in $G'$. 
Let $M$ be such a matching of size $k$.  
The following lemma shows that $M$ yields the
desired edge packing in $G$.
\begin{lemma} \label{lemma:6}
$M$ is an edge packing in $G$ with $|M|=k$.\end{lemma} 
\begin{proof}Suppose $M$ is not an edge packing in $G$.  
Consider any two arbitrary edges $(a_i,s_q)$ and $(a_j,s_{q'})$ from 
$M$.  As $(s_q,s_{q'}) \notin E$ and $s_q\neq s_{q'}$ for any matching $M$, there is only one possibility for $M$ not to be 
an edge packing: $(a_i,a_j)\in E$.  If $(a_i,a_j)\in E$ then $\exists s_p$ in $S$ such that $a_i,a_j\in s_p$.  
But as $U'$ is a set packing, $a_i$ and $a_j$ cannot simultaneously occur in $U'$ yielding a contradiction.  
Hence $M$ is an edge packing in $G$ with $|M|=k$. 
\end{proof}
 Next we will prove the corresponding converse.  
Let $F$ be an edge packing of size $k$ in $G$. Define $P=\{a_i:(a_i,a_j)\in F,i<j\}\cup\{a_j:(a_j,s_q)\in F\}$.  i.e., $P$ 
consists of one endpoint, which is an element of $U$, for every edge in $F$.

\begin{lemma} \label{lemma:7}
$P$ is feasible for $\overline{SC}$ in $T$ and $|P|=k$.\end{lemma} 
\begin{proof}Suppose $P$ is not feasible for $\overline{SC}$ in $T$, then $\exists a_i,a_j\in P$ and $s_q \in S$ such that 
$a_i \in s_q$ and $a_j \in s_q$ and hence by construction of $G$, $(a_i,a_j)\in E(G)$.  Now, $a_i,a_j \in P$ means 
there exists $\alpha,\beta \in U\cup S$
 such that $(a_i,\alpha) \in F$ and $(a_j,\beta)\in F$. As $(a_i,\alpha)$ and $(a_j,\beta)$ are 
two neighbours of $(a_i,a_j)$, $F$ will not be an edge packing of $G$, which is a contradiction.
\end{proof}


\begin{Corollary} \label{cor:2}
 For any $\epsilon > 0$, the edge packing problem ($\overline{EDS}$) is asymptotically not approximable
 within a factor of $|V|^{\frac{1}{2} - \epsilon}$ unless $NP=ZPP$.
\end{Corollary}
\begin{proof}
Let $T=(U,S)$ be an $\overline{SC}$ instance with $|U|=n$ and $|S|=m$.  
By the above reduction we get an edge packing instance $G=(V,E)$ where, $|V|=m+n$.  
Let $K$ be the size of an optimal solution 
for $\overline{SC}$.  By lemma~\ref{lemma:6}, the optimal edge packing in $G$ will also have size $K$.  Hence a $|V|^\alpha$ factor approximation algorithm
for the edge packing problem yields an $\overline{SC}$ solution of size at least $K/(m+n)^\alpha$. 
 By Corollary~\ref{Coro:1}, $\overline{SC}$ is not approximable below a factor 
$m^{\frac{1}{2}-\frac{\epsilon}{2}}$ unless $P=NP$ and $n^{1-\epsilon}$ unless $NP=ZPP$ for any fixed $\epsilon > 0$.  
This implies that, $\alpha\geq \frac{1}{2}(1-\epsilon)max\{\frac{\log m}{\log (m+n)},\frac{2\log n}{\log (m+n)}\}$.  
Thus $\alpha \geq \frac{1}{2}-\epsilon$ for sufficiently large $n$.  
Hence asymptotically a factor of approximation within $|V|^{\frac{1}{2}-\epsilon}$ 
is not achievable for edge packing problem for any $\epsilon>0$, unless $NP=ZPP$.
\end{proof}

\textbf{Part (ii)}\textit{  Reduction from $\overline{SC}$ to $\overline{DS}$:}
Given an
$\overline{SC}$ instance $T=(U,S)$ where, $U=\{a_{1},a_{2},..,a_{n}\}$ and
$S=\{s_{1},s_{2},....s_{m}\}$.  
Construct a graph $G=(V,E)$ such that $V=U \cup S$
and $E=\{(a_i,s_{j}):a_i\in s_{j}\}\cup\{(s_{i},s_{j}):s_{i} \cap s_{j}\neq \emptyset \}$.  Let 
$V'$ be a feasible solution of size $k$ for $\overline{DS}$ in G, then we can construct $U'\subseteq U$ of size $k$ feasible 
for $\overline{SC}$ as follows: for each $a_{i}\in V'$, include $a_{i}$ to $U'$ and for each $s_{j}\in V'$, 
choose any one $a_{i}\in s_{j}$ and include it to $U'$.

\begin{lemma}
 $U'$ is feasible for $\overline{SC}$ in T.
\end{lemma}
\begin{proof} 
 Suppose not.  Then there exists $a_i,a_j$ in $U'$ and $s_p$ in $S$ 
such that $a_i\in s_p$ and $a_j\in s_p$.  This means $(s_p,a_i),(s_p,a_j)\in E$. Thus the sets $N(a_i),N(a_j)$ and 
$N(s_p)$ are not pairwise disjoint in $G$. Now by construction of $V'$, $a_i,a_j\in U'$ implies either $a_i,a_j\in V'$ or 
 $a_i,s_p\in V'$ or  $a_j,s_p\in V'$.  But as $N(a_i),N(a_j)$ and $N(s_p)$ are not pairwise disjoint in $G$, 
all the above cases imply that $V'$ is not a vertex packing in $G$, a contradiction.   
\end{proof}
For the converse, let $U'=\{a_1,a_2,....,a_k\}$ be a feasible solution for $\overline{SC}$ of size $k$ in T.
\begin{lemma}\label{lemma:8}
 $U'$ is a feasible vertex packing in $G$.
\end{lemma}
\begin{proof}
 Suppose not.  Then $\exists a_i,a_j \in U'$ such that $N(a_i)\cap N(a_j)\neq \phi$.  
Hence $\exists s_p \in S$ such that $a_i,a_j \in s_p$.  
But then $U'$ is not a solution for $\overline{SC}$ - a contradiction.  
\end{proof}
   Thus $\overline{SC}$ is reducible to the vertex packing. Arguing exactly as in Corollary \ref{cor:2}
we have:

\begin{Corollary} \label{cor:3}
 For any $\epsilon > 0$, the vertex packing problem ($\overline{DS}$) is asymptotically not approximable
 within a factor of $|V|^ {\frac{1}{2} - \epsilon}$ unless $NP=ZPP$
\end{Corollary}

\textbf{Part (iii)}\textit{ Reduction from $\overline{DS}$ to $\overline{SC}$:}
Given a
graph $G=(V,E)$, we can map the $\overline{DS}$ problem to 
the $\overline{SC}$ problem by constructing an $\overline{SC}$ instance
$T=(U,S)$ such that
 $U=V$ and $S=\{s_{i}=N(i): i\in V\}$.  $S'\subseteq U$ is a feasible solution
for $\overline{SC}$ in $T$ if and only if
for every $v_{i},v_{j} \in S'$, $N(v_{i})\cap N(v_{j})=\phi$ if and only if $S'$
is a vertex packing in $G$. 
\begin{Corollary}\label{lemma:5}
 The vertex packing problem is $\sqrt{|V|}$ approximable.
\end{Corollary}
\begin{proof}
Given a vertex packing problem instance $G=(V,E)$, the above reduction gives an $\overline{SC}$ instance $T=(U,S)$ 
(where $|S|=|V|$) and by Corollary~\ref{Coro:2}, $\overline{SC}$ is $\sqrt{|S|}$ approximable.  
So the vertex packing problem is $\sqrt{|V|}$ approximable.
\end{proof} 

\textbf{Part (iv)}\textit{ Reduction from $\overline{EDS}$ to $\overline{DS}$:}
Given a
graph $G=(V,E)$ (where, $|V|= n$ and $|E|=m$), we can construct a line graph
$G'=(V',E')$ of $G$.   i.e., $V'= E$ and $E'=\{(e,e') : e\in E, 
e'\in N(e)\setminus \{e\} \}$.

Any solution for $\overline{EDS}$ in $G$ will be a solution for $\overline{DS}$ in $G'$. 
 Let $T$ be a feasible edge packing in $G$ and $e_i$ and $e_j$ be any two elements in $T$.  
This implies $N(e_i)\cap N(e_j)=\phi$ which again implies that $T$ is a feasible vertex packing in $G'$.  
Similarly any vertex packing in $G'$ will be an edge packing in $G$. This
completes the proof.

\begin{Corollary} \label{lemma:3}
 The edge packing problem is $\sqrt{|E|}$ approximable.
\end{Corollary}
\begin{proof}
 Given an $\overline{EDS}$ instance $G=(V,E)$  (where, $|V|= n$ and $|E|=m$), we can construct a $\overline{DS}$ instance 
$G'=(V',E')$ (where $|V'|=|E|$) and by Corollary~\ref{lemma:5}, $\overline{DS}$
is $\sqrt{|V|}$ approximable. 
So $\overline{EDS}$ is $\sqrt{|E|}$ approximable.
\end{proof}
This completes the proof of Theorem~\ref{thm:2}.
\end{proof}



\section{Summary and Conclusion}
We summarize the observations proved in the previous sections into the following theorem.
\begin{Theorem}
Let $G=(V,E)$ be any graph.
\begin{enumerate}
 \item  For any $\epsilon > 0$, the edge packing problem ($\overline{EDS}$) is asymptotically not approximable
 within a factor of $|V|^{\frac{1}{2} - \epsilon}$ unless $NP=ZPP$.
\item The edge packing problem is $\sqrt{|E|}$ approximable.
\item For any $\epsilon > 0$, the vertex packing problem ($\overline{DS}$) is asymptotically not approximable
 within a factor of $|V|^{\frac{1}{2} - \epsilon}$ unless $NP=ZPP$.
\item The vertex packing problem is $\sqrt{|V|}$ approximable.
\item The edge packing problem and the vertex packing problem are W[1]-complete under standard parameterization.
\end{enumerate}
\end{Theorem}
\begin{proof}
1, 2, 3 and 4 follows from Corollary~\ref{cor:2}, Corollary~\ref{lemma:3},
Corollary~\ref{cor:3} and Corollary~\ref{lemma:5} 
respectively.  The reductions that proved the equivalence of the vertex packing, 
the edge packing and the set packing problems were parameter preserving under standard parameterization.  Thus all these problems are 
equivalent under FPT reductions.  By Fact~\ref{lema:11}, set packing problem is
W[1]-complete.  Hence 5 follows. 
\end{proof}
In summary, we have shown that $\overline{SC}$, Vertex packing and Edge packing are equivalent to the Set packing problem.  
All reductions presented preserve both 
approximability and parameter, and take only linear time.  
As a consequence edge packing is $\sqrt{|E|}$ approximable and vertex packing is $\sqrt{|V|}$ approximable. 
Both problems are $W[1]$-complete under standard parameterization.  The question of whether $\overline{EDS}$ 
is $O(|V|^{\frac{1}{2}})$ approximable is not known.  Developing good exact algorithms for these problems 
remain open for future investigation.
\bibliographystyle{elsarticle-num}
\bibliography{simple}







\end{document}